\documentclass[runningheads]{llncs}
\usepackage{graphicx}
\usepackage{amsmath}
\usepackage{amssymb}
\usepackage{algorithm}
\usepackage{xspace}
\usepackage{algpseudocode}
\usepackage{hyperref}

\newcommand{\apx}{\textsc{APX}\xspace}
\newcommand{\apxh}{\textsc{APX-hard}\xspace}
\newcommand{\np}{\textsc{NP}\xspace}
\newcommand{\nph}{\textsc{NP-hard}\xspace}
\newcommand{\mkcs}{\textsc{M$k$CS}\xspace}
\newcommand{\msc}{\textsc{MSC}\xspace}

\begin{document}
\title{Approximate Minimum Sum Colorings and Maximum \texorpdfstring{$k$}{k}-Colorable Subgraphs of Chordal Graphs}
\titlerunning{Min. Sum Coloring and Max. \texorpdfstring{$k$}{k}-Colorable Subgraphs of Chordal Graphs}
\author{Ian DeHaan\thanks{Supported by an NSERC Undergraduate Student Research Award held at the University of Alberta.}\inst{1} \and
Zachary Friggstad\thanks{Supported by an NSERC Discovery Grant and Accelerator Supplement.}\inst{2}}
\authorrunning{I. DeHaan and Z. Friggstad}
\institute{Department of Combinatorics and Optimization, University of Waterloo, \email{ijdehaan@uwaterloo.ca} \and
Department of Computing Science, University of Alberta,
\email{zacharyf@ualberta.ca}\\
}
\maketitle              %

\begin{abstract}
We give a $(1.796+\epsilon)$-approximation for the minimum sum coloring problem on chordal graphs, improving over the previous 3.591-approximation by 
Gandhi et al. [2005]. To do so, we also design the first polynomial-time approximation scheme for the maximum $k$-colorable subgraph problem in chordal graphs.%
\end{abstract}

\section{Introduction}

We consider a coloring/scheduling problem introduced by Kubicka in 1989 \cite{ChromSumThesis}.
\begin{definition}
In the \textsc{Minimum Sum Coloring} (\msc) problem, we are given an undirected graph $G = (V,E)$. The goal is to find a proper coloring $\phi : V \rightarrow \{1, 2, 3, \ldots \}$ of vertices with positive integers which minimizes $\sum_{v \in V} \phi(v)$.
In weighted \msc, each vertex $v \in V$ additionally has a weight $w_v \geq 0$ and the goal is then to minimize $\sum_{v \in V} w_v \cdot \phi(v)$.
\end{definition}
Naturally, in saying $\phi$ is a proper coloring, we mean $\phi(u) \neq \phi(v)$ for any edge $uv \in E$.
\msc is often used to model the scheduling of unit-length dependent jobs that utilize shared resources. Jobs that conflict for resources cannot be scheduled
at the same time. The goal in \msc is then to minimize the average time it takes to complete a job.

In contrast with the standard graph coloring problem, where we are asked to minimize the number of colors used, sum coloring is \nph on many simple graph types. 
Even on bipartite and interval graphs, where there are linear time algorithms for graph coloring, \msc remains \apxh \cite{bipartiteSC,intervalSCHard}. 

In \cite{Perfect4Approx}, it was shown that if one can compute a maximum independent set in
any induced subgraph of $G$ in polynomial time, then iteratively coloring $G$ by greedily choosing a maximum independent set of the uncolored nodes each step yields
a 4-approximation for \msc. A series of improved approximations for other graph classes followed, these are summarized in Table \ref{table:results}.
Of particular relevance for this paper are results for perfect graphs and interval graphs. For \msc in perfect graphs, the best approximation is $\mu^{\star} \approx 3.591$, the solution to $\mu \ln \mu = \mu + 1$.
For \msc in interval graphs, the best approximation is $\frac{\mu^{\star}}{2} \approx 1.796$.

\begin{table}
\centering
\def\arraystretch{1.5}

\caption{Known results for sum coloring. The $O^{\star}$-notation hides poly($\log \log n$) factors. Our work appears in bold.}
\label{table:results}

\begin{tabular}{|c|c|c|}
\hline
                & u.b.    & l.b.             \\ \hline
General graphs  & $O^{\star}(n / \log^3 n)$ \cite{Perfect4Approx,GeneralSCApprox}    & $O(n^{1-\epsilon})$ \cite{Perfect4Approx,ColoringHardness} \\ \hline
Perfect graphs  & $3.591$ \cite{PerfectSC} & \apxh  \cite{bipartiteSC}        \\ \hline
Chordal graphs  & ${\bf 1.796+\epsilon}$ & \apxh \cite{intervalSCHard}         \\ \hline
Interval graphs & $1.796$ \cite{IntervalSC} & \apxh \cite{intervalSCHard}         \\ \hline
Bipartite graphs & $27/26$ \cite{bipariteSC2726} & \apxh \cite{bipartiteSC}         \\ \hline
Planar graphs   & PTAS \cite{PlanarSC}    & \nph   \cite{PlanarSC}           \\ \hline
Line graphs     & 1.8298 \cite{LineGraphSC}       & \apxh   \cite{LineGraphHard}           \\ \hline
\end{tabular}

\end{table}
In this paper, we study \msc in chordal graphs. A graph is chordal if it does not contain a cycle of length at least 4 as an induced subgraph. Equivalently, every cycle of length at least 4 has a chord - an edge connecting two non-consecutive nodes on the cycle. Chordal graphs form a subclass of perfect graphs, so we can color them optimally in polynomial time. But \msc itself remains \apxh in chordal graphs \cite{intervalSCHard}, as they generalize interval graphs.
The class of chordal graphs is well studied; linear-time algorithms have been designed to recognize them, to compute maximum independent sets, and to find minimum colorings, among other things. A comprehensive summary of many famous results pertaining to chordal graphs can be found in the excellent book by Golumbic \cite{agt}.
Chordal graphs also appear often in practice; for example Pereira and Palsberg study register allocation problems (which can be viewed as a sort of graph coloring problem) and observe that the interference graphs for about 95\% of the methods in the Java 1.5 library are chordal when compiled with a particular compiler \cite{ChordalImportance}.

Our main result is an improved approximation algorithm for \msc in chordal graphs.
\begin{theorem} \label{scapprox}
    For any constant $\epsilon > 0$, there is a polynomial-time $\frac{\mu^{\star}}{2} + \epsilon \approx 1.796 + \epsilon$ approximation for weighted \msc on chordal graphs.
\end{theorem}
That is, we can approximate \msc in chordal graphs essentially within the same guarantee as for interval graphs.
Prior to our work, the best approximation in chordal graphs was the same as in perfect graphs: a 3.591-approximation by Gandhi et al. \cite{PerfectSC}.

To attain this, we study yet another variant of the coloring problem.
\begin{definition}
In the weighted \textsc{Maximum $k$-Colorable Subgraph} (\mkcs) problem, we are given a graph $G = (V,E)$, vertex weights $w_v \geq 0$, and a positive integer $k$. The goal is to find a maximum-weight subset of nodes $S \subseteq V$ such that the induced subgraph $G[S]$ is $k$-colorable.
\end{definition}
 We also design a polynomial-time approximation scheme (PTAS) for weighted \mkcs in chordal graphs.
\begin{theorem} \label{ptas}
    For any $\epsilon > 0$, there is a $(1 - \epsilon)$-approximation for weighted \mkcs in chordal graphs. 
\end{theorem}
Prior to our work, the best approximation recorded in literature was a $1/2$-approximation by Chakaravarthy and Roy \cite{ChordalColor}. Although one could also get a $(1-1/e)$-approximation by greedily finding and removing a maximum-weight independent set of nodes for $k$ iterations, i.e., the maximum coverage algorithm.

Since Theorem \ref{scapprox} improves the approximation ratio of chordal graphs from the current-best ratio for perfect graphs down to essentially the current-best ratio for interval graphs, one might wonder if we can get an improved approximation for \msc in perfect graphs in general. Indeed, if there was a PTAS for \mkcs in perfect graphs then our approach would yield an improved approximation for perfect graphs. Unfortunately this does not seem possible: we adapt the \np-hardness proof for \mkcs in perfect graphs to show \mkcs is in fact \apx-hard in perfect graphs.

\begin{theorem}\label{thm:apxhard}
\mkcs is \apx-hard in perfect graphs even for $k = 2$.
\end{theorem}

~

\noindent
{\bf Organization}\\
We begin with a high-level discussion of our techniques. Then, Section \ref{sec:lp} presents the proof of Theorem \ref{scapprox} assuming one has a PTAS for \mkcs in chordal graphs. Theorem \ref{ptas} is proven in Section \ref{sec:ptas}. Finally we prove Theorem \ref{thm:apxhard} in Section \ref{sec:apxhard}.

\subsection{Our Techniques}
Our work is inspired by the 1.796-approximation for \msc in interval graphs by Halld\'{o}rsson, Kortsarz, and Shachnai \cite{IntervalSC}. They show that if one has an exact algorithm for \mkcs, then by applying it to values of $k$ from a carefully selected geometric sequence and ``concatenating'' these colorings, one gets a 1.796-approximation. In interval graphs, \mkcs can be solved in polynomial time using a greedy algorithm. We show that a similar result holds: we show Theorem \ref{scapprox} holds in any family of graphs that admit a PTAS for \mkcs. However, we need to use linear programming techniques instead of a greedy algorithm since their approach seems to heavily rely on getting {\em exact} algorithms for \mkcs.

~

\noindent
{\bf \mkcs in Chordal Graphs}\\
In chordal graphs, \mkcs is \textsc{NP-complete}, but it can be solved in $n^{O(k)}$ time \cite{kColorBrute}. We rely on this algorithm for constant values of $k$, so we briefly summarize how it works to give the reader a complete picture of our PTAS. 

Their algorithm starts with the fact that chordal graphs have the following representation. For each chordal graph $G = (V,E)$ there is a tree $T$ with $O(n)$ nodes of maximum degree $3$ plus a collection of subtrees $\mathcal T = \{T_v : v \in V\}$, one for each $v \in V$. These subtrees satisfy the condition that $uv \in E$ if and only if subtrees $T_u$ and $T_v$ have at least one node in common. For a subset $S \subseteq V$, we have $G[S]$ is $k$-colorable if and only if each node in $T$ lies in at most $k$ subtrees from $\{T_v : v \in S\}$. The tree $T$ and subtrees $\mathcal T$ are computed in polynomial time and then a straightforward dynamic programming procedure is used to find the maximum $k$-colorable subgraph. The states of the DP algorithm are characterized by a node $a$ of $T$ and subtrees $\mathcal S \subseteq \mathcal T$ with $|\mathcal S| \leq k$ where each subtree in $\mathcal S$ includes $a$.

Our contribution is an approximation for large values of $k$. It is known that a graph $G$ is chordal if and only if its vertices can be ordered as $v_1, v_2, \ldots, v_n$ such that for every $1 \leq i \leq n$, the set $N^{left}(v_i) := \{v_j : v_iv_j \in E \text{ and } j < i\}$ is a clique. Such an ordering is called a {\bf perfect elimination ordering}. We consider the following LP relaxation based on a perfect elimination ordering. We have a variable $x_v$ for every $v \in V$ indicating if we should include $v$ in the subgraph.
\[ {\bf maximize} \left\{\sum_{v \in V} w_{v} \cdot x_v : x_v + x(N^{left}(v)) \leq k ~\forall~v \in V, x \in [0,1]^V \right\}. \]
The natural $\{0,1\}$ solution corresponding to a $k$-colorable induced subgraph $G[S]$ is feasible, so the optimum LP solution has value at least the size of the largest $k$-colorable subgraph of $G$.

We give an LP-rounding algorithm with the following guarantee.
\begin{lemma}\label{lem:round}
Let $x$ be a feasible LP solution. In $n^{O(1)}$ time, we can find a subset $S \subseteq V$ such that $G[S]$ is $k$-colorable and $\sum_{v \in S} w_v \geq \left(1- \frac{2}{k^{1/3}}\right) \cdot \sum_{v \in V} w_v \cdot x_v$.
\end{lemma}
Theorem \ref{ptas} then follows easily. If $k \leq 8/\epsilon^3$, we use the algorithm from \cite{kColorBrute} which runs in polynomial time since $k$ is bounded by a constant. Otherwise, we run our LP rounding procedure.

~

\noindent
{\bf Linear Programming Techniques for \msc}\\
We give a general framework for turning approximations for weighted \mkcs into approximations for \msc.
\begin{definition}
We say that an algorithm for weighted \mkcs is a $(\rho, \gamma)$ approximation if it always returns a $\gamma \cdot k$ colorable subgraph with vertex weight at least $\rho \cdot OPT$, where $OPT$ is the maximum vertex weight of any $k$-colorable subgraph.
\end{definition}
For Theorem \ref{scapprox}, we only need to consider the case $\rho = 1-\epsilon$ and $\gamma = 1$. Still, we consider this more general concept since it is not any harder to describe and may have other applications.

We prove the following, where $e$ denotes the base of the natural logarithm.
\begin{lemma}\label{lem:lp}
    Suppose there is a $(\rho, \gamma)$ approximation for weighted \mkcs on some class of graphs. Then, for any $1 < c < \min(e^2, \frac{1}{1-\rho})$, there is a $\frac{\rho \cdot \gamma \cdot (c + 1)}{2 \cdot (1- (1-\rho)\cdot c)\cdot \ln c}$-approximation for \msc for graphs in the same graph class.
\end{lemma}
Our main result follows by taking $\gamma = 1$ and $\rho = 1-\epsilon$. For small enough $\epsilon$, we then choose $c^* \approx 3.591$ to minimize the expression, resulting in
an approximation guarantee of at most $1.796$.

Roughly speaking, we prove Lemma \ref{lem:lp} by considering a time-indexed configuration LP relaxation for latency-style problems. Configuration LPs have been considered for \msc in other graph classes, such as line graphs \cite{LineGraphSC}. The configurations used in previous work have variables for each independent set. We use a stronger LP that has variables for each $k$-colorable subgraph for each $1 \leq k \leq n$. 

Our configuration LP was inspired by one introduced by Chakrabarty and Swamy for the \textsc{Minimum Latency Problem} (a variant of the \textsc{Travelling Salesperson Problem}) \cite{ChakrabartySwamyMinLat}, but is tailored for our setting. For each ``time'' $k \geq 1$ we have a family of variables, one for each $k$-colorable subgraph, indicating if this is the set of nodes that should be colored with integers $\leq k$. This LP can be solved approximately using the $(\rho, \gamma)$-approximation for \mkcs, and it can be rounded in a manner inspired by \cite{ChakrabartySwamyMinLat,PostSwamyMinLat}.

Note that Theorem \ref{ptas} describes a $(1-\epsilon, 1)$-approximation for \mkcs for any constant $\epsilon > 0$. If we had a $(1, 1+\epsilon)$-approximation then the techniques in \cite{IntervalSC} could be easily adapted to prove Theorem \ref{scapprox}. But these techniques don't seem to apply when given \mkcs approximations that are inexact on the number of nodes included in the solution.

\section{An LP-Based Approximation Algorithm for \msc} \label{sec:lp}
As mentioned earlier, our approach is inspired by a time-indexed LP relaxation for latency problems introduced by Chakrabarty and Swamy \cite{ChakrabartySwamyMinLat}. Our analysis follows ideas presented by Post and Swamy
who, among other things, give a 3.591-approximation for the \textsc{Minimum Latency Problem} \cite{PostSwamyMinLat} using a configuration LP.

\subsection{The Configuration LP} 

For a value $k \geq 0$ (perhaps non-integer), $\mathcal{C}_k$ denotes the vertex subsets $S \subseteq V$ such that $G[S]$ can be colored using at most $k$ colors.
For integers $1 \leq k \leq n$ and each $C \in \mathcal{C}_k$, we introduce a variable $z_{C, k}$ that indicates if $C$ is the set of nodes colored with the first $k$ integers.
We also use variables $x_{v, k}$ to indicate vertex $v$ should receive color $k$. We only need to consider $n$ different colors since no color will be ``skipped'' in an optimal solution.

\begin{alignat}{3}
\text{\bf minimize:} & \quad & \sum_{v \in V} \sum_{k=1}^n w_v \cdot k \cdot x_{v, k} \tag{{\bf LP-MSC}} \label{LP1} \\
\text{\bf subject to:} && \sum_{k=1}^n x_{v, k} = \quad & 1 \quad && ~\forall~v \in V \label{primal1} \\
&& \sum_{C \in \mathcal{C}_k} z_{C, k} \leq \quad & 1 && ~\forall~ 1 \leq k \leq n \label{primal2}  \\
&& \sum_{C \in \mathcal{C}_k : v \in C} z_{C, k} \geq \quad & \sum_{k' \leq k} x_{v, k'} && ~\forall~ v \in V, 1 \leq k \leq n \label{primal3} \\
&& x, z \geq \quad & 0 \notag
\end{alignat}

Constraint \eqref{primal1} says each vertex should receive one color, constraint \eqref{primal2} ensures we pick just one subset of vertices to use the first $k$ colors on, and 
constraint \eqref{primal3} enforces that each vertex colored by a value less than or equal to $k$ must be in the set we use the first $k$ colors on.

Recall that this work is not the first time a configuration LP has been used for \msc. In \cite{LineGraphSC}, the authors consider one that has a variable $x_{C,k}$ for every independent set $C$, where the variable models that $C$ is the independent set used for color $t$. Our approach allows us to prove better bounds via LP rounding, but it has the stronger requirement that in order to (approximately) solve our LP, one needs to (approximately) solve the \mkcs problem, rather than just the maximum independent set problem.

Let $OPT$ denote the optimal \msc cost of the given graph and $OPT_{LP}$ denote the optimal cost of \eqref{LP1}. Then $OPT_{LP} \leq OPT$ simply because the natural $\{0,1\}$ solution corresponding to $OPT$ is feasible for this LP.

At a high level, we give a method to solve this LP approximately by using the algorithm for \mkcs to approximately separate the constraints of the dual LP, which is given as follows.
\begin{alignat}{3}
\text{\bf maximize:} & \quad & \sum_{v \in V} \alpha_v - \sum_{k=1}^n \beta_k \tag{{\bf DUAL-MSC}} \label{LD1} \\
\text{\bf subject to:} && \alpha_v \leq \quad & w_v \cdot k + \sum_{\hat{k} = k}^n \theta_{v, \hat{k}} \quad && ~\forall~ v \in V, 1 \leq k \leq n \label{dual1}  \\
&& \sum_{v \in C} \theta_{v, k} \leq \quad & \beta_k && ~\forall~ 1 \leq k \leq n, C \in \mathcal{C}_k \label{dual2} \\
&& \beta, \theta \geq \quad & 0 \label{dual3}
\end{alignat}

Note \eqref{LD1} has polynomially-many variables. We approximately separate the constraints in the following way.
For values $\nu \geq 0, \rho \leq 1, \gamma \geq 1$, let $\mathcal D(\nu; \rho, \gamma)$ denote the following polytope:
\begin{equation}\label{eqn:dual}
\left\{(\alpha, \beta, \theta) : \eqref{dual1}, \eqref{dual3}, \sum_{v \in C} \theta_{v,k} \leq \beta_k ~\forall~1 \leq k \leq n ~\forall~ C \in \mathcal C_{\gamma \cdot k}, \sum_v \alpha_v - \frac1\rho \cdot \sum_k \beta_k \geq \nu \right\}
\end{equation}

\begin{lemma}\label{lem:sep}
If there is a $(\rho,\gamma)$-approximation for \mkcs, there is also a polynomial-time algorithm $\mathcal A$ that takes a single value $\nu$ plus values $(\alpha, \beta, \theta)$ for the variables of \eqref{LD1} and always returns one of two things:
\begin{itemize}
\item A (correct) declaration that $(\alpha, \beta/\rho, \theta) \in \mathcal D(\nu; 1, 1)$.
\item A constraint from $\mathcal D(\nu; \rho, \gamma)$ that is violated by $(\alpha, \beta, \theta)$.
\end{itemize}
\end{lemma}
\begin{proof}

First, check that \eqref{dual1}, \eqref{dual3}, and $\sum_v \alpha_v - \frac{1}{\rho} \cdot \sum_k \beta_k \geq \nu$ hold. If not, we already found a violated constraint.
Otherwise, for each $k$ we then run the \mkcs $(\rho, \gamma)$-approximation on the instance with vertex weights $\theta_{v,k}, v \in V$. If this finds a solution (i.e. a $(k \cdot \gamma)$-colorable subgraph) with weight exceeding $\beta_k$, we return the corresponding violated constraint. Otherwise, we know that the maximum possible weight of a $k$-colorable subgraph is at most $\beta_k/\rho$. If the latter holds for all $k$, then $(\alpha, \beta/\rho, \theta) \in \mathcal D(\nu; 1, 1)$.
\end{proof}

Lemma 3.3 from \cite{ChakrabartySwamyMinLat} takes such a routine and turns it into an approximate LP solver. The following is proven in the exact same manner where we
let ${\bf LP}^{(\rho,\gamma)}$ be the same as \eqref{LP1}, except $\mathcal C_k$ is replaced by $\mathcal C_{\gamma \cdot k}$ in both \eqref{primal2} and \eqref{primal3} and the right-hand side of \eqref{primal2} is replaced by $1/\rho$.

\begin{lemma} \label{lem:apxsolve}
For fixed constant rational values $\rho \geq 1, \gamma \leq 1$, given a $(\rho, \gamma)$-approximation for \mkcs, we can find a feasible solution $(x,z$) to ${\bf LP}^{(\rho,\gamma)}$ with cost at most $OPT_{LP}$ in polynomial time.
\end{lemma}

\begin{proof}
Our proof is nearly identical to that in \cite{ChakrabartySwamyMinLat} (see Section 3 in their paper). For completeness, we provide our own argument.

First, we trivially know $\mathcal D(0; 1; 1) \neq \emptyset$. 
We also know that the optimal primal solution has value at most $1 + n \cdot \sum_{v} w_v$, and so by weak LP duality, we have that ${\mathcal D(1 + n \cdot \sum_v w_v; \rho, \gamma) = \emptyset}$.

We consider a binary search over the range $[0, 1 + n \cdot \sum_v w_v]$ and maintain the invariant that for the current search window $[\ell, \mu]$ that $\mathcal D(\ell; 1, 1) \neq \emptyset$ and $\mathcal D(\mu; \rho, \gamma) = \emptyset$. Initially this is true by the above discussion.

Now consider some $\nu$ in the middle of a given range $[\ell, \mu]$. We run the ellipsoid method over variables $(\alpha, \beta, \theta)$ and invoke Lemma \ref{lem:sep} for each such tuple encountered.
If it returns a declaration that $(\alpha, \beta/\rho, \theta) \in \mathcal D(\nu; 1, 1)$, we update the lower end of the binary search range, i.e. $\ell := \nu$. Otherwise, after a polynomial number of iterations the ellipsoid method will generate an infeasible collection of constraints from $\mathcal D(\nu; \rho, \gamma)$ in which case we update the upper end of the binary search range, i.e. $\mu := \nu$.

Eventually the binary search algorithm reduces the range to be simply $[\nu, \nu+\epsilon]$ for some appropriate value $\epsilon$ we will describe soon. Note that $\mathcal D(\nu'; 1, 1) = \emptyset$ for all $\nu' > OPT_{LP}$ so this final range must have $\nu \leq OPT_{LP}$.

Consider the constraints, say $\mathcal H$, that were generated to certify $\mathcal D(\nu + \epsilon; \rho, \gamma) = \emptyset$. These constraints correspond to a polynomial-size set of variables of ${\bf LP}^{(\rho, \gamma)}$. We claim that for an appropriately-small choice of $\epsilon$ that there is a feasible solution to ${\bf LP}^{(\rho, \gamma)}$ using only these variables (i.e. setting all others to 0) with cost at most $OPT_{LP}$. If so, since the corresponding LP has polynomial size then we can find it ourselves by solving the resulting LP.

To see this, first add all of \eqref{dual1}, \eqref{dual3}, and the ``objective function'' constraint $\sum_v \alpha_v - 1/\rho \sum_k \beta_k \geq \nu + \epsilon$ to $\mathcal H$ and notice $\mathcal H$ remains inconsistent (i.e. the inequalities cannot all be satisfied). A routine application of Farkas' lemma on this system of constraints then shows the restriction of ${\bf LP}^{(\rho, \gamma)}$ to just the variables corresponding to constraints in this modified version of $\mathcal H$ has a feasible solution of value at most $\nu + \epsilon$. Thus, ${\bf LP}^{(\rho, \gamma)}$ has an extreme point solution $(x,z)$ with value $\nu'$ at most $\nu + \epsilon \leq OPT_{LP} + \epsilon$. We claim, in fact, that $\nu' \leq OPT_{LP}$ by choosing $\epsilon$ appropriately small.

By standard extreme point analysis, there is some integer $D$ whose bit complexity is polynomial in the input (i.e. in the number of bits used to describe $G, w, \rho, \gamma$) such that all extreme points of ${\bf LP}^{(\rho, \gamma)}$, in particular the value of $(x,z)$, have their objective value having denominator $\leq D$. Similarly there is such an integer $D'$ upper bounding the denominator in the value of any extreme point to \eqref{LP1}, in particular this applies to $OPT_{LP}$.

Thus, if we let $\epsilon = 1/(D \cdot D')$ then the value $\nu'$ of the extreme point solution $(x,z)$ to ${\bf LP}^{(\rho, \gamma)}$ satisfies $\nu' \leq OPT_{LP} + 1/(D \cdot D')$. Assume, by way of contradiction, that $\nu' > OPT_{LP}$. Then $\nu$' and $OPT_{LP}$ are two different values with $|OPT_{LP} - \nu'| < 1/(D \cdot D')$. But the difference between distinct numbers with denominators at most $D$ and $D'$ respectively must be at least $1/(D \cdot D')$, a contradiction. So $\nu' \leq OPT_{LP}$ as required, i.e. $(x,z)$ is the desired solution.
\end{proof}

\subsection{The Rounding Algorithm and Analysis}

The rounding algorithm is much like that in \cite{IntervalSC} in that it samples $k$-colorable subgraphs for various values of $k$ in a geometric sequence and concatenates these colorings to get a coloring of all nodes.
For convenience, let $z_{C, k} = z_{C, \lfloor k \rfloor}$ for any real value $k \geq 0$.

\begin{algorithm}[H]
\caption{MSCRound(G)}\label{alg:msc}
\begin{algorithmic}
\State find a solution $(x, z)$ to ${\bf LP}^{(\rho,\gamma)}$ with value $\leq OPT$ using Lemma \ref{lem:apxsolve}
\State if necessary, increase $z_{\emptyset, k}$ until $\sum_{C \in \mathcal C_{\gamma \cdot k}} z_{C,k} = 1/\rho$ for each $k$
\State let $1 < c < \min(e^2, 1/(1-\rho))$ be a constant we will optimize later
\State let $h = c^{\Gamma}$ be a random offset where $\Gamma$ is sampled uniformly from $[0, 1)$
\State $j \gets 0$
\State $k \gets 0$ \Comment{The next color to use}
\While{$G \neq \emptyset$}
\State $k_j \gets h \cdot c^j$
\State $k'_j \gets \min\{n, \lfloor k_j \rfloor\}$
\State choose $C$ randomly from $\mathcal{C}_{\gamma \cdot k'_j}$ with probability according to the LP values $z_{C', k'_j} \cdot \rho $ for $C' \in \mathcal{C}_{\gamma \cdot k'_j}$
\State color $C$ with $\lfloor \gamma \cdot k'_j\rfloor$ colors, call the color classes $C_1, C_2, \ldots, C_{ \gamma \cdot k'_j}$
\State randomly permute the color classes, let $C'_1, C'_2, \ldots, C'_{ \gamma \cdot k'_j}$ be the reordering
\State finally, assign nodes in $C'_i$ color $k + i$ for each $1 \leq i \leq \lfloor \gamma \cdot k'_j \rfloor$ in the final solution
\State $k \gets k + \lfloor \gamma \cdot k'_j \rfloor$
\State $G \gets G - C$
\State $j \gets j + 1$
\EndWhile
\end{algorithmic}
\end{algorithm}
Note that nodes colored during iteration $j$ get assigned colors at most $\gamma \cdot (k_0 + k_1 + \ldots + k_j)$ and the expected color of such a node is at most $\gamma \cdot (k_0 + k_1 + \ldots + k_{j-1} + (k_j+1)/2)$.
The number of iterations is $O(\log n)$ because each vertex will appear in each $\gamma \cdot n$ coloring, as this is the largest color considered in ${\bf LP}^{(\rho, \gamma)}$.

We note that despite our approach following the main ideas of the algorithm and analysis for minimum latency given in \cite{PostSwamyMinLat}, there are some key details that change. In \cite{PostSwamyMinLat}, each iteration of the algorithm produces a tree, which is then doubled and shortcutted to produce a cycle with cost at most double the tree. While we randomly permute the colors in our coloring, they randomly choose which direction to walk along the cycle. For a tree with cost $k$, this gives an expected distance of $k$ for each node. We save a factor of $2$ because we do not have a doubling step, but our average color is $\frac{k+1}{2}$ as opposed to $\frac{k}{2}$. Some extra work is required in our analysis to account for the extra $\frac{1}{2}$ on each vertex.

Let $p_{v, j}$ be the probability that vertex $v$ is not colored by the end of iteration $j$. For $j < 0$, we use $p_{v,j} = 1$ and $k_j = 0$.
Finally, for $v \in V$, let $\phi(v)$ denote the color assigned to $v$ in the algorithm.

The following is essentially Claim 5.2 in \cite{PostSwamyMinLat}, with some changes based on the differences in our setting as outlined above.
\begin{lemma}\label{lemma:fixed}
For a vertex $v$,
\[{\bf E}[\phi(v) | h] \leq \frac{\gamma}{2} \cdot \frac{c+1}{c-1} \cdot \sum_{j \geq 0} p_{v, j-1} \cdot (k_j - k_{j-1}) + \gamma \cdot\left( \frac{1}{2} - \frac{h}{c-1} \right).\]
\end{lemma}
\begin{proof}
There are at most $\gamma \cdot k_j$ colors introduced in iteration $j$. They are permuted randomly, so any vertex colored in iteration $j$ has color, in expectation, at most $\gamma \cdot (k_j+1)/2$ more than all colors used in previous iterations.
That is, the expected color of $v$ if colored in iteration $j$ is at most
\begin{eqnarray*} \gamma\cdot \left(k_0 + k_1 + \ldots + k_{j-1} + \frac{k_j+1}{2}\right) & \leq &
\gamma\cdot \left( h \cdot \left(\frac{c^j-1}{c-1} + \frac{c^j}{2} \right) + \frac{1}{2} \right) \\
& = & \gamma\cdot \left( \frac{k_j}{2} \cdot \frac{c+1}{c-1} + \frac{1}{2} - \frac{h}{c-1} \right),
\end{eqnarray*}
where we have used $k_i = h \cdot c^i$ and summed a geometric sequence.

The probability $v$ is colored in iteration $j$ is $p_{v,j-1} - p_{v,j}$, so the expected color of $v$ is bounded by
\[ \frac{\gamma}{2} \cdot \frac{c+1}{c-1} \cdot \left( \sum_{j \geq 0}  (p_{v,j-1} - p_{v,j}) \cdot k_j \right) + \gamma \cdot\left( \frac{1}{2} - \frac{h}{c-1} \right). \]
By rearranging, this is what we wanted to show.
\end{proof}

For brevity, let $y_{v,j} = \sum_{k \leq k_j} x_{v,k}$ denote the LP coverage for $v$ up to color $k_j$.
The next lemma is essentially Claim 5.3 from \cite{PostSwamyMinLat}, but the dependence on $\rho$ is better in our context\footnote{We note \cite{PostSwamyMinLat} does have a similar calculation in a single-vehicle setting of their problem whose dependence is more like that in Lemma \ref{lemma:bound}. They just don't have a specific claim summarizing this calculation that we can reference.}.
\begin{lemma}\label{lemma:bound}
For any $v \in V$ and $j \geq 0$, we have $p_{v, j} \leq (1-y_{v,j}) \cdot \rho + (1 - \rho) \cdot p_{v, j-1}$.
\end{lemma}
\begin{proof}
If $v$ is not covered by iteration $j$, then it is not covered in iteration $j$ itself and it is not covered by iteration $j-1$,
which happens with probability
\begin{eqnarray*}
p_{v,j-1} \cdot \left(1 - \sum_{C \in \mathcal C_{\gamma \cdot k_j} : v \in C} \rho \cdot z_{C, k_j}\right) & \leq & p_{v,j-1} \cdot (1-\rho \cdot y_{v,j}) \\
& = & p_{v,j-1} \cdot \rho \cdot (1-y_{v,j}) + p_{v,j-1} \cdot (1-\rho).
\end{eqnarray*}
Note that the first inequality follows from constraint \eqref{primal3} and the definition of $y_{v, j}$. The lemma then follows by using $p_{v,j-1} \leq 1$ and $y_{v,j} \leq 1$ to justify dropping $p_{v,j-1}$ from the first term.
\end{proof}

From these lemmas, we can complete our analysis. Here, for $v \in V$, we let $col_v = \sum_{k = 1}^n k \cdot x_{v,k}$ denote the fractional color of $v$, so the cost of $(x,z)$ is $\sum_{v \in V} w_v \cdot col_v$.
The following lemma is essentially Lemma 5.4 in \cite{PostSwamyMinLat} but with our specific calculations from the previous lemmas.

\begin{lemma}
For any $v \in V$, we have ${\bf E}[\phi(v)] \leq \frac{\rho \cdot \gamma \cdot (c + 1)}{2 \cdot (1- (1-\rho)\cdot c)\cdot \ln c} \cdot col_v$.
\end{lemma}
\begin{proof}
For brevity, let $\Delta_j = k_j - k_{j-1}$.
We first consider a fixed offset $h$. Let $A = \sum_{j \geq 0} p_{v, j-1} \cdot \Delta_j$ and recall, by Lemma \ref{lemma:fixed}, that the expected color of $v$ for a given $h$ is at most $\frac{\gamma}{2} \cdot \frac{c+1}{c-1} \cdot A + \gamma (\frac{1}{2} - \frac{h}{c-1})$.

Note $\Delta_j = c \cdot \Delta_{j-1}$ for $j \geq 2$ and $\Delta_0 + \Delta_1 = c \cdot \Delta_0$.
So from Lemma \ref{lemma:bound},
\begin{eqnarray*}
A & \leq & \sum_{j \geq 0} \rho \cdot (1-y_{v,j}) \cdot \Delta_j + (1-\rho) \sum_{j \geq 0} p_{v, j-2} \cdot \Delta_j \\
& = & \sum_{j \geq 0} \rho \cdot (1-y_{v,j}) \cdot \Delta_j + c \cdot (1-\rho) \cdot A.
\end{eqnarray*}
Rearranging and using $c < 1/(1-\rho)$, we have that
\[ A \leq \frac{\rho}{1-c \cdot (1-\rho)} \cdot  \sum_{j \geq 0} (1-y_{v,j}) \cdot \Delta_j. \]

For $1 \leq k \leq n$, let $\sigma(k)$ be $k_j$ for the smallest integer $j$ such that $k_j \geq k$.  Simple manipulation and recalling $y_{v,j} = \sum_{k \leq k_j} x_{v,j}$
shows $\sum_{j \geq 0} (1-y_{v,j}) \cdot \Delta_j = \sum_{k=1}^n \sigma(k) \cdot x_{v,k}$. 

The expected value of $\sigma(k)$ over the random choice of $h$, which is really over the random choice of $\Gamma \in [0,1)$, can be directly calculated as follows
where $j$ is the integer such that $k \in [c^j, c^{j+1})$.
\begin{eqnarray*}
{\bf E}_h[\sigma(k)] & = & \int_{0}^{\log_c k-j} c^{\Gamma+j+1} d\Gamma + \int_{\log_c k - j}^1 c^{\Gamma+j} d\Gamma \\
 & = & \frac{1}{\ln c} (c^{\log_c k + 1} - c^{j+1} + c^{j+1} - c^{\log_c k}) = \frac{c-1}{\ln c} \cdot k.
\end{eqnarray*}
We have just shown ${\bf E}_h[\sum_{j \geq 0} (1-y_{v,j}) \cdot \Delta_j] = \frac{c-1}{\ln c} \sum_{k \geq 0} k \cdot x_{v,k} = \frac{c-1}{\ln c} \cdot col_v$. So, we can now bound the unconditional color ${\bf E}_h[\phi(v)]$ using our previous lemmas.
\begin{eqnarray*}
{\bf E}_h[\phi(v)] & = & \frac{\gamma}{2} \cdot \frac{c+1}{c-1} \cdot {\bf E}_h[A] + \gamma \left(\frac{1}{2} - {\bf E}_h[h] / (c-1)\right) \\
& \leq & \frac{\rho \cdot \gamma \cdot (c + 1)}{2 \cdot (1 - (1 - \rho) \cdot c) \cdot \ln c} \cdot col_v + \gamma \left( \frac{1}{2} - {\bf E}_h[h] / (c-1) \right) \\
& = & \frac{\rho \cdot \gamma \cdot (c + 1)}{2 \cdot (1 - (1 - \rho) \cdot c) \cdot \ln c} \cdot col_v + \gamma \left( \frac{1}{2} - \frac{1}{\ln c} \right) \\
& \leq & \frac{\rho \cdot \gamma \cdot (c + 1)}{2 \cdot (1 - (1 - \rho) \cdot c) \cdot \ln c} \cdot col_v
\end{eqnarray*}
The first equality and inequality follow from linearity of expectation and known bounds on ${\bf E}[\phi(v) | h]$ and $A$. The second equality follows from the fact that ${\bf E}_h[h] = \int_{0}^1 c^\Gamma d\Gamma = \frac{c-1}{\ln c}$, and the last inequality is due to the fact that $c < e^2$ by assumption.

\end{proof}

To finish the proof of Lemma \ref{lem:lp}, observe the expected vertex-weighted sum of colors of all nodes is then at most
\[ \frac{\rho \cdot \gamma \cdot (c + 1)}{2 \cdot (1- (1-\rho)\cdot c)\cdot \ln c} \cdot \sum_{v \in V} w_v \cdot col_v \leq \frac{\rho \cdot \gamma \cdot (c + 1)}{2 \cdot (1- (1-\rho)\cdot c)\cdot \ln c} \cdot OPT.\]
Theorem \ref{scapprox} then follows by combing the $(1-\epsilon,1)$ \mkcs approximation (described in the next section) with this \msc approximation, choosing $c \approx 3.591$, and ensuring $\epsilon$ is small enough so $c < 1/\epsilon$.

We note Algorithm \ref{alg:msc} can be efficiently derandomized. First, there are only polynomially-many offsets of $h$ that need to be tried. That is, for each $k_j$, we can determine the values of $h$ that would cause $\lfloor \gamma \cdot k_j \rfloor$ to change and try all such $h$ over all $j$. Second, instead of randomly permuting the color classes in a $\gamma \cdot k_j$-coloring, we can order them greedily in non-increasing order of total vertex weight.

\section{A PTAS for \mkcs in Chordal Graphs} \label{sec:ptas}

We first find a perfect elimination ordering of the vertices $v_1, v_2, \ldots, v_n$.
This can be done in linear time, e.g., using lexicographical breadth-first search \cite{agt}.
Let $N^{left}(v) \subseteq V$ be the set of neighbors of $v$ that come before $v$ in the ordering, so $N^{left}(v) \cup \{v\}$ is a clique.
Recall that we are working with the following LP. The constraints we use exploit the fact that a chordal graph is $k$-colorable if and only if all left neighbourhoods of its nodes in a perfect elimination ordering have size at most $k-1$.

\begin{alignat}{3}
\text{maximize:} & \quad & \sum_{v \in V} w_v \cdot x_v \tag{{\bf K-COLOR-LP}} \label{lpcolor} \\
\text{subject to:} && x_v + x(N^{left}(v)) \leq \quad & k \quad && \forall~v \in V \label{boundk}  \\
&& x \in \quad & [0,1]^V \notag
\end{alignat}

Let $OPT_{LP}$ denote the optimal LP value and $OPT$ denote the optimal solution to the problem instance. Of course, $OPT_{LP} \geq OPT$ since the natural $\{0,1\}$ integer solution corresponding to a $k$-colorable subgraph of $G$
is a feasible solution.

We can now give a rounding algorithm as follows.

\begin{algorithm}[H]
\caption{MCSRound($G$, $k$)}
\label{alg:mkcs}
\begin{algorithmic}

\State let $0 \leq f(k) \leq 1$ be a value we will optimize later
\State find a perfect elimination ordering $v_1, v_2, \ldots, v_n$ for $G$
\State let $x$ be an optimal feasible solution to \eqref{lpcolor}
\State form $S'$ by adding each $v \in V$ to $S'$ independently with probability $(1-f(k)) \cdot x_v$.
\State $S \gets \emptyset$
\For{$v \in \{v_1, v_2, \ldots, v_n\}$}
    \State if $v \in S'$ and $S \cup \{v\}$ is $k$-colorable, add $v$ to $S$
\EndFor
\State \Return $S$
\end{algorithmic}
\end{algorithm}

\subsection{Analysis}

Observe that when we consider adding some $v \in S'$ to $S$, $S \cup \{v\}$ is $k$-colorable if and only if $|S \cap N^{left}(v)| \leq k-1$. This is easy to prove by noting that the restriction of a perfect elimination ordering of $G$ to a subset $S$ yields a perfect elimination ordering of $G[S]$. Because we consider the nodes $v$ according to a perfect elimination ordering of $G$, by adding $v$ the only possible left-neighbourhood of a node that could have size $\geq k$ is $N^{left}(v)$ itself.

We bound the probability that we select at least $k$ vertices from $N^{left}(v)$. The second moment method is used so that derandomization is easy.
Let $Y_u$ indicate the event that $u \in S'$. 
Then  ${\bf E}[Y_u^2]  = {\bf E}[Y_u] = (1 - f(k)) \cdot x_u$.
Fix some vertex $v$. Let $Y = \sum_{u \in N^{left}(v)} Y_u$. 
By constraint \eqref{boundk}, we have
\[ {\bf E}[Y] 
= \sum_{u \in N^{left}(v)} (1-f(k)) \cdot x_u 
\leq (1 - f(k)) \cdot k. \] 
And since each $Y_u$ is independent, we have again by constraint \eqref{boundk} that
\[ {\bf Var}[Y] = \sum_{u \in N^{left}(v)} {\bf Var}[Y_u] = \sum_{u \in N^{left}(v)} \left( {\bf E}[Y_u^2] - {\bf E}[Y_u]^2 \right) \leq \sum_{u \in N^{left}(v)} {\bf E}[Y_u^2] \leq k. \]
We are interested in
\[ {\bf Pr}[Y \geq k] \leq {\bf Pr}[|Y-E[Y]| \geq f(k) \cdot k]. \] 
By Chebyshev's inequality,
\[
{\bf Pr}[|Y-E[Y]| \geq f(k) \cdot k] \leq \frac{{\bf Var}[Y]}{f(k)^2 \cdot k^2} \leq \frac{k}{f(k)^2 \cdot k^2} = \frac{1}{f(k)^2 \cdot k}.
\]
From this, we find that the probability we actually select vertex $v$ is at least
\[ {\bf Pr}[Y_v \wedge (Y \leq k-1)] = {\bf Pr}[Y_v] \cdot {\bf Pr}[Y \leq k-1] \geq (1-f(k)) \cdot x_v \cdot \left(1-\frac{1}{f(k)^2 \cdot k}\right). \]
The first equality is justified because $Y$ only depends on $Y_u$ for $u \neq v$, so these two events are independent.

Choosing $f(k) = k^{-1/3}$ results in $v \in S$ with probability at least $x_v \cdot (1 - 2 \cdot k^{-1/3})$.
By linearity of expectation, the expected value of $S$ is at least $(1 - 2 \cdot k^{-1/3}) \cdot \sum_{v \in V} w_v \cdot x_v$. 

The PTAS for \mkcs in chordal graphs is now immediate. For any constant $\epsilon > 0$, if $k \geq 8/\epsilon^3$, then we run our LP rounding algorithm to get a $k$-colorable subgraph with weight at least $(1-\epsilon) \cdot OPT_{LP}$.
Otherwise, we run the exact algorithm in \cite{kColorBrute}, which runs in polynomial time since $k$ is bounded by a constant.

It is desirable to derandomize this algorithm so it always finds a solution with the stated guarantee. This is because we use it numerous times in the approximate separation oracle for \eqref{LD1}. Knowing it works all the time 
does not burden us with providing concentration around the probability we successfully approximately solve ${\bf LP}^{(\rho,\gamma)}$ as in Lemma \ref{lem:apxsolve}.
We can derandomize Algorithm \ref{alg:mkcs} using standard techniques since it only requires that the variables $Y_u, u \in V$ be pairwise-independent (in order to bound ${\bf Var}[Y]$).

\section{\apx-Hardness for \mkcs in Perfect Graphs}\label{sec:apxhard}
It is natural to wonder if \msc admits a better approximation in perfect graphs. Unfortunately, the techniques we used to get better approximations for chordal graphs do not extend immediately to perfect graphs. In \cite{PerfectBipartiteSubgraphHard}, Addario-Berry et al. showed \mkcs is \nph in a different subclass of perfect graphs than chordal graphs. Their proof reduces from the maximum independent set problem and it is easy to see it shows \mkcs is \apxh in the same graph class if one reduces from bounded-degree instances of maximum independent set as we now show.

\begin{proof}[of Theorem \ref{thm:apxhard}]
In \cite{PerfectBipartiteSubgraphHard}, a polynomial-time reduction is given that will produce a graph $H$ from a given graph $G$ such that the maximum 2-colorable subgraph in $H$ has size exactly $3 \cdot |V(G)| + 2 \cdot |E(G)| + \alpha(G)$ where $\alpha(G)$ is the size of the largest independent set in $G$. Consider applying this reduction when $G = (V,E)$ is a cubic graph (i.e., a simple graph where every vertex has degree exactly 3). 

Since $G$ is a cubic graph, $2 \cdot |E(G)| =  \sum_{v \in V} \deg(v) = 3 \cdot |V(G)|$ so the maximum independent set size is in fact $6 \cdot |V(G)| + \alpha(G)$. Also observe $\alpha(G) \geq |V|/4$ since any $n$-node graph with maximum degree $d$ has an independent set of size $n/(d+1)$ that can be constructed by greedily picking nodes not adjacent to any previously-picked nodes.

Thus, for any $0 < c < 1$ if there is a $c$-approximation for the maximum 2-colorable subgraph of $H$ then we would have computed a value $z$ such that
\[ 6 \cdot |V(G)| + \alpha(G) \geq z \geq c \cdot (6 \cdot |V(G)| + \alpha(G)). \]
So $(z - c \cdot 6 \cdot V(G))/(25 - 24 \cdot c) \in [\alpha(G)/(25 - c \cdot 24), \alpha(G)]$, i.e. we have a $1/(25-24 \cdot c)$-approximation for the value of $\alpha(G)$. But the \textsc{Maximum Independent Set} problem in cubic graphs is \apx-hard meaning
a $c$-approximation for the maximum $2$-colorable subgraph problem cannot (barring surprising developments in complexity theory) exist for constants $c$ sufficiently close to $1$.

\end{proof}

\section{Conclusion}
Our approach, or a refinement of it, may succeed in getting a good approximation for \msc in perfect graphs if one has good constant approximations for \mkcs in perfect graphs. Note that \mkcs can be approximated within $1-1/e \approx 0.6321$ in perfect graphs simply by using the {\em maximum coverage} approach. That is, for $k$ iterations, we greedily compute a maximum independent set of nodes that are not yet covered. This is not sufficient to get an improved \msc approximation in perfect graphs using Lemma \ref{lem:lp}. That is, Lemma \ref{lem:lp} can yield improve \mkcs approximation if we get a sufficiently-good $(\approx0.704)$ approximation for \mkcs. As a starting point, we ask if there is a $\rho$-approximation for \mkcs in perfect graphs for some constant $\rho > 1-1/e$.

\bibliographystyle{splncs04}

\bibliography{sc}

\end{document}